\documentclass[a4paper, 11pt, UKenglish]{article}
\usepackage[
  backend=biber,
  doi=true,
  url=false,
  isbn=false,
  eprint=true
]{biblatex}
\usepackage{hyperref}
\hypersetup{
    colorlinks,
    linkcolor={blue!50!black},
    citecolor={black},
    urlcolor={blue!80!black}
}

\usepackage{babel}
\usepackage[utf8]{inputenc}
\usepackage[T1]{fontenc}
\usepackage{geometry}

\usepackage{amsmath, amsthm, amssymb}
\usepackage{enumitem}
\usepackage{todonotes}
\usepackage{thmtools}
\usepackage{thm-restate}
\usepackage{mathtools}
\usepackage[capitalise]{cleveref}
\usepackage{xspace}
\usepackage{setspace}
\usepackage{tikz}
\usepackage{multicol}
\usepackage{comment}

\declaretheorem[name=Theorem, numberwithin=section]{theorem}
\declaretheorem[name=Definition, style=definition, sibling=theorem]{definition}
\declaretheorem[name=Lemma, sibling=theorem]{lemma}

\declaretheorem[name=Corollary, sibling=theorem]{corollary}

\declaretheorem[name=Question, sibling=theorem]{question}

\newcommand{\mainonly}[1]{}

\crefname{claim}{Claim}{Claims}
\crefname{case}{case}{cases}
\crefname{prop}{Property}{Properties}

\newcommand{\set}[1]{\{#1\}}
\newcommand{\setof}[2]{\{#1\mid #2\}}

\renewcommand{\subset}{\subseteq}

\newcommand{\from}{\colon}

\newcommand{\C}{\mathcal{C}}

\renewcommand{\P}{\mathcal{P}}

\renewcommand{\O}{\mathcal{O}}

\newcommand{\N}{\mathbb{N}}
\DeclareMathOperator{\mw}{mw}
\DeclareMathOperator{\tww}{tww}

\newcommand{\dist}{\mathrm{dist}}

\newcommand{\Ball}{\mathrm{Ball}}
\newcommand{\VC}{\mathrm{VCdim}}

\let\le\leqslant

\let\leq\leqslant
\let\geq\geqslant

\let\emptyset\varnothing

\let\preceq\preccurlyeq   % 

\def\MIIS{\textsc{Max Dist-$2$ Independent Set}\xspace}
\def\MDS{\textsc{Min Dominating Set}\xspace}
\def\rMIIS{\textsc{Max Dist-$2r$ Independent Set}\xspace}
\def\rMDS{\textsc{Min Dist-$r$ Dominating Set}\xspace}
\def\ii{2-independence\xspace}
\def\iis{distance-$2$ independent set\xspace}
\def\diratio{domination-to-2-independence ratio\xspace}

\title{Constant-factor approximation of maximum distance-2 independent set in graphs of bounded merge-width}

\author{Maël Dumas\footnotemark \vspace{0.5cm}  \\ 
University of Warsaw, Institute of Informatics, Poland}

\date{}

\begin{document}
\maketitle
\begin{abstract}
    We give a constant-factor approximation algorithm for \MIIS in graphs of bounded radius-$2$ merge-width.
    The same result holds for \MDS from [Bonamy and Geniet, 2025], [Chan et al., SODA  '12].
    Both approximation algorithms are LP-based, showing that the \diratio is bounded in graphs of bounded radius-$2$ merge-width. Moreover, this result is tight in the sense that the ratio can be unbounded in graphs of bounded radius-$1$ merge-width.
\end{abstract}

\renewcommand{\thefootnote}{\fnsymbol{footnote}}
\footnotetext{\footnotemark
Maël Dumas was supported by the ERC project BUKA (n° 101126229) (during his employment in BUKA) and project BOBR (during his employment in BOBR) that received funding from the European Research Council (ERC) under the European Union’s Horizon 2020 research and innovation programme (grant agreement No. 948057).}
\renewcommand{\thefootnote}{\arabic{footnote}}
\setcounter{footnote}{0}

% \paragraph*{Acknowledgement} Maël Dumas was supported by the ERC project BUKA (n° 101126229) (during his employment in BUKA) and project BOBR (during his employment in BOBR) that received funding from the European Research Council (ERC) under the European Union’s Horizon 2020 research and innovation programme (grant agreement No. 948057).

\section{Introduction}

For a graph $G$, a \iis is a set $I\subset V(G)$ whose vertices are at pairwise distance more than $2$, or equivalently, that the closed neighbourhoods of vertices in $I$ are pairwise disjoint.
The maximum size of such a set is the \emph{\ii number} $\alpha_2(G)$, and the corresponding optimisation problem is \MIIS.
Dually, the \MDS{} problem asks for a minimum-size set $D \subseteq V(G)$ such that every vertex of $G$ is in $D$ or adjacent to a vertex of $D$. The minimum size of such a set is the \emph{domination number} $\gamma(G)$. These two values are related by the inequality $\alpha_2(G) \leq \gamma(G)$.

% The size of a maximum distance-$2$ independent set is denoted $\alpha_2(G)$ and the . 
% The minimum dominating set (\MDS) problem ask for a minimum-size set $D\subset V(G)$ of vertices such that every vertex of $G$ is either in $D$ or adjacent to a vertex in $D$, the size of a solution to this problem is the dominating number of $G$, denoted $\gamma(G)$.
% It is well known that $\alpha_2(G) \leq \gamma(G)$ and that these problems have natural integer linear programs  formulation such that their relaxation are dual.

Both problems are NP-complete: \textsc{Set Cover} reduces to \MDS{} and \textsc{Set Packing} reduces to \MIIS{}, and both are among Karp's 21 NP-complete problems~\cite{Karp72}.
On general graphs, \MDS{} is approximable with factor $\ln n$~\cite{Johnson74,Lovasz75} and this is essentially tight unless $\mathrm{P} = \mathrm{NP}$~\cite{DinurS14}.
The problem \MIIS{} is NP-hard to approximate within $n^{1/2-\varepsilon}$ for any $\varepsilon > 0$, even in bipartite and chordal graphs~\cite{EtoGM14}, and is APX-hard in cubic graphs, though it admits a PTAS on planar graphs~\cite{EtoILM22}.
 
On restricted graph classes, however, constant-factor approximations in polynomial-time become attainable.
For \MDS{}, there are constant-factor approximations for graph classes of bounded degeneracy~\cite{LenzenW10,BansalU17} and of linear neighbourhood complexity~\cite{CGKS12}. The latter encompass a broad range of graph classes, such as planar, bounded degree, bounded expansion and bounded twin-width.
For \MIIS{}, there is are constant-factor approximations for cubic graphs~\cite{EtoILM17}, graph classes of bounded expansion~\cite{Dvorak13} %, more precisely, for graph classes of bounded weak colouring number $2$,  
and bounded twin-width~\cite{BGKTW24}. Note that for the latter, the constant-factor approximation requires that a decomposition of small twin-width is given as part of the input, however such decomposition is not known to be computable in polynomial-time.

Merge-width is a family of graph parameters recently introduced by Dreier and Toru\'{n}czyk~\cite{merge-width}. In this paper we are specifically interested in graph classes of bounded radius-$2$ merge-width, which include graphs classes such as planar, bounded degree, bounded expansion and bounded twin-width.

Notably, it was shown by Bonamy and Geniet~\cite{bonamy2025chiboundedness} that graph classes of bounded radius-$2$ merge-width have linear neighbourhood complexity, hence \MDS admits constant-factor approximation in these classes by the previously mentioned result of Chan et al.~\cite{CGKS12}.

\paragraph*{Contributions} The main result is the following. 

\begin{theorem}\label{thm:main}
    \MIIS{} admits a constant-factor approximation in graphs of bounded radius-$2$ merge-width.
\end{theorem}

See \cref{thm:miis_approx} for explicit bounds. Our approximation is simple greedy LP-relaxation-based algorithm, which  pick vertices whose radius-$2$ ball in the graph minimize the weight of a fractional \iis. 
The main technical lemma is that 
the weight of these balls are bounded by a function of the radius-$2$ merge width.
Additionally, we provide a tighter bound for bounded twin-width, improving upon the result of Bonnet et al.~\cite{BGKTW24} (and providing a truly polynomial-time constant-factor approximation).

\paragraph*{Related work}
The relaxation of the integer linear program formulations of the two problems are LP-duals of each other, that is $\gamma^*(G)= \alpha_2^*(G)$.
A natural question is to ask is if ratio $\gamma/\alpha_2$, called the \emph{\diratio}, is small, and has been widely investigated. The ratio is bounded in many graph classes such as asteroidal triple-free~\cite{BCGY25}, bounded expansion~\cite{Dvorak13,Dvorak19} or bounded twin-width~\cite{BGKTW24}; and unbounded in bipartite 3-degenerate graphs~\cite{Dvorak13} and split~\cite{BCGY25}.

Both the constant-factor approximation of Chan et al.~\cite{CGKS12} for \MDS and ours for \MIIS are LP-based, in particular, the integrality gap of these problems are bounded. Since the LP relaxation of these problems are dual, it shows that the \emph{\diratio} is bounded in graphs of bounded radius-$2$ merge-width.

Natural generalisations of the considered problems are \rMDS and \rMIIS. They also admit constant-factor approximations in graphs of bounded merge-width since these graph classes are closed under taking power. Note however that the radius of the merge-width may depend on $r$, see \cref{sec:concl} for more details.

\section{Preliminaries} \label{sec:prelim}

\paragraph*{Graphs}
%For two sets $A,B$, we write $AB \coloneqq\setof{ab}{a\in A,b\in B,a\neq b}$, where $ab$ denotes the unordered pair $\{a,b\}$, and $\binom{A}{2} \coloneqq AA$.
Graphs are simple, undirected, and finite, that is, a graph $G$ consists of a finite set $V(G)$ of vertices and a set $E(G)\subseteq \binom{V(G)}{2}$ of edges.
The edge-complement of~$G$ is~$\overline{G}$, with $V(\overline{G})\coloneqq V(G)$ and $E(\overline{G}) \coloneqq \binom{V(G)}{2} \setminus E(G)$.

By $\dist_G(x,y)$ we denote the length of a shortest path between $x$ and~$y$ in $G$ %(or $\infty$ if no path exists)
, and denote $\Ball^r_G(v) \coloneqq\setof{w\in V(G)}{\dist_G(v,w)\le r}$. By $N_G(v)$ we denote the open neighbourhood of~$v$ in $G$, that is, the set of vertices adjacent to $v$; and by $N_G(X)$ we denote the set of vertices adjacent to $X\subseteq V(G)$, excluding $X$ itself.

For two sets $A,B$, we write $AB \coloneqq\setof{ab}{a\in A,b\in B,a\neq b}$, where $ab$ denotes the unordered pair $\{a,b\}$. Two sets $A,B\subseteq V(G)$ are \emph{complete} if $AB\subseteq E(G)$, and \emph{anticomplete} if $AB\cap E(G)=\emptyset$.
The \emph{restriction} of a partition $\P$ of $V(G)$ to a set $S\subseteq V(G)$ is defined as $$\P\restriction_S = \{P \cap S \mid P\in \P\}\setminus \{\emptyset\}.$$

\paragraph*{Merge-width}
Let~$\P$ be a partition of the vertices $V$ of a graph~$G$, and $R \subseteq \binom{V}{2}$ be a set of pairs of vertices of~$G$.
By $\dist_R(x,y)$ and $\Ball_R^r(x)$ we mean the corresponding notions in the graph $(V,R)$.
We say that~$\P$ is \emph{homogeneous modulo $R$} (in $G$) if for any parts $A,B \in \P$ (possibly $A=B$), either all edges or all non-edges are in $R$, that is for every pairs $ab,a'b'\in AB\setminus R$, we have that $ab \in E(G)$ if and only if $a'b' \in E(G)$.
The \emph{radius-$r$ width} of $(\P,R)$~is
\[ \max_{v \in V} |\P\restriction_{\Ball^r_R(v)}|. \]

\begin{definition}[{\cite[Def. 3.1]{merge-width}}]
	A \emph{merge sequence} for a graph~$G$ is a sequence $$(\P_1,R_1),\dots,(\P_m,R_m)$$ where
	\begin{enumerate}
		\item $\P_1\preccurlyeq \P_2\preccurlyeq\ldots\preccurlyeq \P_m$ is a sequence of ever coarser partitions of $V(G)$ with~$\P_1$ the partition into singletons and $\P_m$ the partition with one part, 
		\item $R_1 \subseteq \dots \subseteq R_m \subseteq \binom{V(G)}{2}$ is a monotone sequence of set of pairs of vertices, and
		\item $\P_t$ is homogeneous modulo $R_t$, for $t=1,\ldots,m$. 
	\end{enumerate}
	The \emph{radius-$r$ width} of this merge sequence is the maximum radius-$r$ width of $(\P_{t},R_{t+1})$,
	for  $1 \le t < m$.
	Finally, the \emph{radius-$r$ merge-width} of~$G$, denoted by~$\mw_r(G)$, is the minimum radius-$r$ width of a merge sequence for~$G$.
\end{definition}
Note that the mismatched indices in $(\P_t,R_{t+1})$ are intentional, and forbid one from merging many parts and adding many resolved pairs all at once when going from step~$t$ to~$t+1$.
A graph class~$\C$ has \emph{bounded merge-width} there is a function $f$ such that $\mw_r(G)<f(r)$  for every $r\in\N$ and $G\in \C$.

In the definition of merge-width, the sequence of partitions $\P_1\preceq\ldots\preceq \P_m$ may be equivalently required to be a \emph{maximal chain of partitions} of $V(G)$ (i.e. $|\P_i| = |\P_{i+1}|+1$).
Indeed, given any merge sequence $(\P_1,R_1),\dots,(\P_m,R_m)$,
one can transform it into a merge sequence whose sequence of partitions is a maximal chain of partitions as follows: if $\P_i=\P_{i+1}$ for some $1\le i<m$, then we can drop the pair $(\P_{i+1},R_{i+1})$ in the sequence, and if $\P_i\prec \P\prec \P_{i+1}$, then we may insert the pair $(\P,R_{i+1})$ into the sequence. Those operations do not increase the radius-$r$ width of the merge sequence.

\paragraph*{VC-dimension and neighbourhood complexity}
A set~$A$ of vertices of a graph $G$ is \emph{shattered} 
if $\setof{N_G(v)\cap A}{v\in V(G)}=2^A$.
% that is, $N^A=2^A$.
The \emph{VC-dimension} of~$G$ is the maximum size of a shattered subset of $V(G)$.
The \emph{neighbourhood complexity function} (or \emph{shatter function}) is defined by:
$$ \pi_G(m) \coloneqq \max_{\substack{A\subseteq V(G),\\|A|\le m}}\left|\setof{N_G(v)\cap A}{v\in V(G)}\right|.$$
%This extends to a class~$\C$ of graphs as $\pi_\C(m) \coloneqq \max_{G \in \C} \pi_G(m)$.

Note the trivial bound $\pi_G(m) \le 2^m$. The fundamental Sauer-Shelah-Perles lemma \cite{sauer72,shelah72} states that this bound is polynomial in graphs of bounded VC-dimension.

\begin{lemma}[Sauer-Shelah-Perles lemma] \label{sauer_shelah_lemma}
	Let $G$ be a graph of VC-dimension~$d$. Then $$\pi_G(m) \le \O(m^d)\qquad\text{for all $m\in\N$.}$$
\end{lemma}
For graphs of bounded radius-$2$ merge-width,
Bonamy and Geniet~\cite[Thm. 1.5]{bonamy2025chiboundedness}
proved that the neighbourhood complexity function is linear.

\begin{theorem}\label{thm:nbd_complexity}
	Any graph $G$ with $\mw_2(G)= k$ has neighbourhood complexity
	$$\pi_G(m)\leq k2^{k+2}\cdot m\qquad\text{for all $m\in\N$.}$$
\end{theorem}

Toruńczyk~{\cite[Thm.~5.24]{flip-width}} showed that the VC-dimension of a graph $G$ is linear in its radius-$1$ flip-width. 
The proof can be easily adapted for radius-$1$ merge-width (see also~\cite[Lemma~7.20]{merge-width}).
\begin{theorem}\label{thm:vc_mw1}
	Any graph $G$ satisfies $\VC(G) \le \O(\mw_1(G))$.
\end{theorem}

\paragraph*{Duality} Our main result makes use of the following combinatorial notion.
For a graph $G$ and subsets of vertices $X, Y \subseteq V(G)$, a set $S_{XY}$ is \emph{dual} for $(X,Y)$ if:
\begin{itemize}
	\item $S_{XY} \subseteq X$ and $S_{XY}$ dominates $Y$, i.e. $Y\subseteq N_G(S_{XY})$, or
	\item $S_{XY} \subseteq Y$ and $S_{XY}$ anti-dominates $X$, i.e. $X\subseteq N_{\overline{G}}(S_{XY}).$
\end{itemize}
A graph $G$ is said to have \emph{a duality of order} $d\in \mathbb{N}$ if every pair of subsets of $V(G)$ has a dual of size at most $d$.
The following, together with \cref{thm:vc_mw1} shows that the duality of a graph $G$ is linearly upper bounded by $\mw_1(G)$.
\begin{theorem}[{\cite[Thm.~E.2]{flip-width}}]\label{thm:dual_vc}
	Any graph $G$ has a duality of order $\O(\VC(G))$.
\end{theorem}

\paragraph*{Linear Program formulation} The domination number and the \ii number can be defined as optima of the following linear programs: \\
\begin{minipage}{0.45\textwidth}
\begin{align*}
   % & \MDS\\
  & \text{minimize} \sum_{x \in V(G)} w(x) \\
  & \text{s.t.} \sum_{y \in N[x]} w(y)\geq 1,  \forall x \in V(G) \\
  &  w(x) \in \{0,1\}, \quad \forall x \in V(G)
\end{align*}
\end{minipage}%  
\hfill
\begin{minipage}{0.45\textwidth}
\begin{align*}
    %& \MIIS\\
  & \text{maximize } \sum_{x \in V(G)} w(x) \\
  & \text{s.t.} \sum_{y \in N[x]} w(y)\leq 1,  \forall x \in V(G) \\
  &  w(x) \in \{0,1\}, \quad \forall x \in V(G)
\end{align*}
\end{minipage}%
\medskip\\
The relaxations of these linear programs are obtained by replacing the constraints $w(x)\in\set{0,1}$ by $w(x)\in [0,1]$.
The optima of the relaxations are called \emph{fractional domination number} and \emph{fractional \ii number}, denoted by $\gamma^*(G)$ and $\alpha_2^*(G)$, respectively.
Since the relaxed linear programs are duals, it follows that $\gamma^*(G)=\alpha_2^*(G)$.

For \MDS{} we will use a result of Chan et al.~\cite[Thm.~1.1]{CGKS12}
(see also \cite[Thm. 1]{kupavskii}).
They show a strong approximation result for instances of set covers with small \emph{shallow cell complexity} (SCC). The latter can be upper bounded in terms of the shatter function of the set system, and adapting their result to our setting, we get the following.

\begin{theorem}[\cite{CGKS12}]\label{thm:approx_mds} 
	Fix a non-decreasing function $f\from \N\to \N_{+}$. 
    \MDS{} admits a randomized polynomial time $\O\left(\log f(m)\right)$-approximation algorithm for graphs $G$ with $\pi_G(m)\leq m\cdot f(m)$ for all $m\in\N$. Moreover $$\gamma(G) \le \O\Big(\log f(m)\cdot \gamma^*(G)\Big).$$
\end{theorem}

Combining \cref{thm:approx_mds}, which states that classes of bounded radius-$2$ merge-width have linear neighborhood complexity~\cite{bonamy2025chiboundedness}, with \cref{thm:nbd_complexity}, we get that in such classes \MDS{} has bounded integrality gap.

\begin{theorem}[follows from \cite{bonamy2025chiboundedness,CGKS12}] \label{thm:mds_gap}
For any graph $G$ with $\mw_2(G)\leq k$, we have
$$\gamma(G) \le \O(k\cdot \gamma^*(G)).$$
\end{theorem}

\section{Approximation of \MIIS{}} \label{sec:approx}

The constant-factor approximation for \MIIS{} is a simple greedy algorithm, iteratively picking vertices whose ball of radius $2$ in $G$ have minimum weight according to a fractional \ii weight function. The crucial part is that these balls have  weight bounded by a function of radius-$2$ merge-width.
Note that for a fractional \ii weight function $w \from V(G) \to [0,1]$  of $G$, for every $x\in V(G)$, we have $w(N[x]) \leq 1$. This function is optimal if $w^*(V(G))= \alpha_2^*(G)$. 
% For a graph $G$, we say that a function $w \from V(G) \to [0,1]$ is \emph{fractional \ii} of $G$ if for every $x\in V(G)$, we have $w(N[x]) \leq 1$. Such a function is optimal if $w^*(V(G))= \alpha_2^*(G)$. 

\begin{lemma}\label{lem:mw_ball}
    For any graph $G$ of radius-$2$ merge-width $k$ and duality of order $d$, for any fractional \ii weight function $w^*$ of $G$, there is a vertex $x\in V(G)$ such that $$w^*(\Ball^2_G(x)) = O(kd^2).$$
\end{lemma}
\begin{proof}
Let $G$ be a graph of duality of order $d$ and radius-$2$ merge-width $k$ witnessed by a merge-sequence $(\P_1, R_1), \dots , (\P_n , R_n)$. Moreover, assume w.l.o.g. that $|\P_i| = |\P_{i+1}|+1$. 
Let $w^* : V(G) \to \mathbb{R}$ be any fractional \ii weight function of $G$.
Let $i$ be the smallest index such that there is a part $X\in \P_i$ with $w^*(X)>d$. We will show that for any $x\in X$, we have $w^*(\Ball_G^2(x)) \le \O(kd^2)$, hence proving the lemma.

By the choice of $X$, since $|\P_i| = |\P_{i+1}|+1$ we have $w^*(X) \leq 2d$ and for any other part of $P\in \P_i$, we have $w^*(P) \leq d$. 
Define now the following sets:    
    \begin{align*}
    X_R      &= N_{R_{i+1}}(x)
    &\qquad X_{R,R} &= N_{R_{i+1}}(X_R)
    &\qquad X_{R,B} &= N_G(X_R) \setminus X_{R,R} \\
    X_B      &= N_G(x) \setminus X_R
    &\qquad X_{B,R} &= N_{R_{i+1}}(X_B)
    &\qquad X_{B,B} &= N_G(X_B) \setminus X_{B,R}
    \end{align*}
    Observe that \[\Ball_G^2(x) \subseteq \set x\cup X_R \cup X_B \cup X_{R,R} \cup X_{R,B} \cup X_{B,R} \cup X_{B,B},\] hence it is enough to bound the weight of each of these sets.
    First, since $X_R \cup X_{R,R} \subseteq \Ball_{R_{i+1}}^2(x)$ by construction, these sets intersect at most $k$ parts of $\P_i$ and $w^*(X_R \cup X_{R,R}) \leq(k+1)d$. %O(kd) 
    
    Consider now $X_{R,B}$. Let $u\in X_{R,B}$; then there is some $v \in N_G(u) \cap X_R$. Let $P \in  \P_i\restriction_{X_R}$ be such that $v\in P$. Then $P \subseteq N_G(u)$, as otherwise $u\in N_{R_{i+1}}(P) \subseteq X_{R,R}$. 
    Therefore, any vertex  $u\in X_{R,B}$ is complete to some part of $\P_i\restriction_{X_R}$.
    Moreover for $P,Q\subseteq V(G)$, if $P$ and $Q$ are complete to each other, then $w^*(P)\leq 1$ and $w^*(Q) \leq 1$.
    This implies that $w^*(N_G(P) \cap X_{R,B}) \leq 1 $ for every part $P\in \P_i\restriction_{X_R}$, hence $w^*(X_{R,B}) \leq k$.

    To conclude the proof we will show that there exists $S \subseteq X$  such that $X_B \subseteq N_{R_{i+1}}(S)$ and $|S| \leq d$. 
    This implies that $w^*((X_B \cup X_{B,R}) \setminus X) \leq kd^2$ %O(kd^2)
    since $\Ball_{R_{i+1}}^2(S) \supseteq X_B \cup X_{B,R}$ intersect at most $kd$ parts, each of them having weight at most $d$. Furthermore, as $|\P_i\restriction_{X_B}|\leq kd$, using the same arguments as for $X_{R,B}$, we have  $w^*(X_{B,B}) \leq kd$.
    Recall that $d$ is the order of duality of $G$, that is for $X_B$ and $X$:
    \begin{itemize}
    \item there is a set $S\subset X_B$, $|S|\le d$, such that $X \subseteq N_G(S)$, or 
    \item there is a set $S\subset X$, $|S|\le d$, such that $X_B \subseteq N_{\overline{G}}(S)$.
\end{itemize}
    Since $w^*(N_G(v)) \leq 1, v\in V(G)$ and $w^*(X)> d$, there is no set $S\subseteq V(G)$ of size at most $d$ such that $X \subseteq N_G(S)$.
    Hence there is a set $S \subseteq X$ such that $X_B \subseteq N_{\overline{G}}(S)$ and $|S|\leq d$. 
    Moreover, since $N_{R_{i+1}}(x) \cap X_B = \emptyset$ and $\P_i$ is homogeneous modulo $R_{i+1}$, we have that for $v\in X, u\in X_B$ if $uv\notin E(G)$, then $uv\in R_{i+1}$, hence $X_B \subseteq N_{R_{i+1}}(S)$.
\end{proof}

The previous lemma can be adapted to obtain a bound that depends only on radius-$3$ merge-width, which may be smaller in some instances.

\begin{lemma} \label{lem:mw3_ball}
        For any graph $G$ of radius-$3$ merge-width $k$, for any fractional \ii weight function $w^*$ of $G$, there is a vertex $x\in V(G)$ with $w^*(x)>0$, such that $$w^*(\Ball^2_G(x)) = O(k).$$
\end{lemma}
\begin{proof}
    %In \cref{lem:mw_ball}, if the radius-$3$ merge-width is at most $k$, then we an show that $w^*(\Ball_G^2(x)) \le \O(k)$. 
    The proof follows the same structure as that of \cref{lem:mw_ball}. The main difference here is that we consider the first index such that a part has weight more than $2$ instead of $d$. Let $X$ be the part of weight more than $2$ (and less than $4$). The analysis for $X_R$, $X_{R,R}$ and $X_{R,B}$ is the same.

    The crucial observation is that every vertex in $X_B$ share a common non-neighbour in $X$ as $w^*(X)>2$ and $w^*(N_G[u])\leq 1$.
    As any non-edge between $X_B$ and $X$ is in $R_{i+1}$, it follows that for every  $v \in X_B$,  $X_B \subseteq \Ball_{R_{i+1}}^2(v)$ and furthermore $X_B \cup X_{B,R} \subseteq \Ball_{R_{i+1}}^3(v)$ and hence intersect $k$ parts. The analysis of the weight for every other subset of $\Ball_G^2(x)$ is the same as in the proof of \cref{lem:mw_ball}. Therefore, as every part except $X$ has weight less than $2$, it follows that $w^*(\Ball_G^2(x)) \leq \O(k)$.
\end{proof}

    From \cref{lem:mw_ball,lem:mw3_ball} we can design a greedy constant-factor approximation algorithm for \MIIS{} in graph classes of bounded radius-$2$ merge-width.

    \begin{theorem}\label{thm:miis_approx}
    \MIIS{} admits a  $$\O(\min\{k_2d^2,k_3\})\text{-approximation}$$ for graphs $G$ where $k_2\geq \mw_2(G)$, $k_3\geq \mw_3(G)$ and $d$ is the order of duality of $G$. Moreover, $ \alpha_2^*(G) \leq \O(\min\{k_2d^2,k_3\}\cdot \alpha_2(G)\big).$
    \end{theorem}
    % \begin{theorem}\label{thm:miis_approx}
    %     For any graph $G$, \MIIS{} can be $\O(\min\{k_2d^2,k_3\})$-approximated %$\Omega(\max\{\frac{1}{k_2d^2},\frac{1}{k_3}\})$-approximation 
    % , where $d$ is the order of duality of $G$, $k_2=\mw_2(G)$ and $k_3=\mw_3(G)$. Moreover $$ \alpha_2^*(G) \leq \O(\min\{k_2d^2,k_3\}\cdot \alpha_2(G)\big).$$
    % \end{theorem}
    \begin{proof}
    Let $w^*$ % : V(G) \to [0,1]$ 
    be an optimal fractional \ii weight function of $G$ obtained by solving the relaxation of the linear program for $\MIIS$.
    We will construct a \iis $I$ of $G$ by the following procedure:
    while $w^*(V(G))>0$, pick a vertex $x\in V(G)$ with $w^*(x)>0$ such that $w^*(\Ball_G^2(x))$ is minimum; add $x$ to $I$ and set $w^*(y)$ to $0$ for every $y\in \Ball_G^2(x)$.

    First, note that $I$ is a \iis of $G$, indeed, whenever a vertex is picked, its current weight is non-zero, this implies that its distance to vertices previously added to $I$ is at least $3$.
    Moreover, whenever a vertex is added to $I$, the total weight of $G$ is decreased by at most $O(\min\{k_2d^2,k_3\})$ by \cref{lem:mw_ball,lem:mw3_ball}. Hence, $\O(|I|\cdot \min\{k_2d^2,k_3\})  \geq \alpha_2^*$, concluding the proof.
    \end{proof}

\cref{thm:miis_approx} shows that in graph class of bounded radius-$2$ merge-width, integrality gap for \MDS{} is bounded, since in any graph the duality of order is bounded by the radius-$1$ merge-width (\cref{thm:vc_mw1,thm:dual_vc}).
Since $\gamma^*(G) = \alpha_2^*(G)$, it follows from \cref{thm:mds_gap,thm:miis_approx} that the \diratio is bounded in graph of bounded radius-$2$ merge-width.

\begin{corollary}
	For any graph $G$ of duality of order $d$ and with $\mw_2(G)\leq k$, we have $$\gamma(G)  \le \O(k^2d^2 \cdot \alpha_2(G)).$$
\end{corollary}

The \emph{twin-width} of a graph $G$, denoted by $\tww(G)$, is the minimum integer $k$ such that there exists a maximal chain 
$$\P_1\prec \P_2\prec \dots\prec \P_n$$
of partitions of $V(G)$ (called a \emph{contraction sequence} in this context) such that for every $t\in[n]$, every part $P\in\P_t$ is non-homogeneous to at most $k$ parts of $\P_t\setminus \{P\}$.
We can adapt the approximation algorithm of \cref{thm:miis_approx} to the setting of twin-width and get an improvement of the approximation factor over the one for \MIIS{} of Bonnet et al.~\cite{BGKTW24}. Moreover, their approximation required a contraction sequence given as part of the input.
This also reaches the $\O(\tww(G)^2)$ integrality gap implied by Bonamy et al.~\cite{BCGY25}.

\begin{corollary} \MIIS{} admits a $\O(k^2)$-approximation for graphs $G$ where $\tww(G)\leq k$.
\end{corollary}

\begin{proof}
    Let $\P_1\prec\ldots\prec\P_n$ be a contraction sequence witnessing that $\tww(G)\leq k$.
    Observe that whenever a part $P$ has weight more than one, no vertex of $G$ can be complete to it.
    In particular, all vertices of $G$ fully adjacent to this part have to be in one of the at most $k$ parts non-homogeneous with $P$. It follows that for a step $\P_i$ of a contraction sequence, if only one part $X$ has weight more than one, then $w^*(\Ball_G^2(x)) \le \O(k^2)$ for any $x \in X$, and the conclusion follows from the algorithm defined in \cref{thm:miis_approx}.
    
    More precisely, let $\mathcal{X}_R$ be the set of parts non-homogeneous with $X$; $\mathcal{X}_{R,R}$ the set of parts non-homogeneous with at least one part of $\mathcal{X}_R$; and $X_{R,B}$ is the set of vertices fully adjacent to at least one part of $\mathcal{X}_R$.
    Then it is clear that $\Ball_G^2(x) \subseteq P\cup \bigcup \mathcal{X}_R \cup \bigcup \mathcal{X}_{R,R} \cup X_{R,B} $. Since $|\mathcal{X}_{R,R}| \leq k^2$ and that for each part $\mathcal{X}_R$ the weight of the vertices fully adjacent to it is at most one, it follows that $w^*(\Ball_G^2(x)) \le \O(k^2)$.

\end{proof}

\section{Related problems and open questions} \label{sec:concl}
Natural generalisations of the two problems investigated in this paper are \rMDS{} and \rMIIS{}. In these problems we look respectively for a minimum-size set dominating every vertex at distance at most $r$, and a maximum-size set such that every vertex is at distance at most $2r$ from at most one vertex from this set.     
For $r\in\mathbb{N}$, the $r$-power $G^r$ of a graph $G$ is the graph  obtained by adding an edge between every pair of vertices at distance at most $r$. It is straightforward to see that solving \MDS{} on $G^r$ provides a solution of \rMDS{} on $G$, and similarly for the independent set variant.
    
Moreover, it is known that the radius-$2$ merge-width of $G^r$ depends only on the radius-$f(r)$ merge-width of $G$ for some function $f$~\cite[Theorem 1.12]{merge-width}\footnote{This theorem states that graph classes of bounded merge-width are closed under first-order interpretation. The graph $G^r$ can be interpreted from $G$ by a first-order formula $\varphi(u,v)$ that is satisfied if $\dist_G(u,v) \leq r$.}. 
Therefore, \rMDS{} and \rMIIS{} can be constant-factor approximated and have bounded integrality gap in graphs of bounded radius-$f(r)$ merge-width.

\medskip

It can be asked if the results of this paper can be extended to broader classes of graphs, for example the ones of bounded radius-$1$ merge-width.
Dvor{\'{a}}k~\cite{Dvorak13} %~\cite[Theorem 1.14.]{BCGY25} 
observed that there are graphs that are $3$-degenerate and for which the \diratio is unbounded. 
Since $k$-degenerate graphs have radius-$1$ merge-width at most $k+2$ \cite[Theorem 7.3]{merge-width}, the \diratio is also unbounded in graphs of bounded radius-$1$ merge-width.
In particular, the counter-example used shows that the integrality gap \MIIS{} in $3$-degenerate graphs can be unbounded. While this does not rule out a constant-factor approximation for \MIIS{} in graphs of bounded radius-$1$ merge-width, it seems unlikely. However, this leaves the following
question for \MDS{}.

\begin{question}
Does \MDS{} admit a constant-factor approximation in graphs of bounded radius-$1$ merge-width? Does it have bounded integrality gap?
\end{question}

Chan et al.~\cite{CGKS12} show that \MDS{} admits a constant factor approximation in graphs of linear neighbourhood complexity (\cref{thm:approx_mds}). Our main result extends the approximability at constant factor of \MIIS{} to graphs of bounded radius-$2$ merge-width, we ask if this can be pushed further.

\begin{question}
    Does \MIIS{} admit a constant-factor approximation in graphs of linear neighbourhood complexity?
\end{question}

% \bibliographystyle{plainurl}
% \bibliography{refs}
\printbibliography

% \bibliographystyle{cas-model2-names}

% % Loading bibliography database
% \bibliography{refs}

\end{document}